\documentclass[11pt]{article}

\usepackage[margin=1in]{geometry} % to modify margins

\usepackage[T1]{fontenc}
\usepackage{textcomp}
\usepackage{palatino}
\usepackage{mathpazo}
\usepackage{stmaryrd}

\usepackage{hyperref}
\hypersetup{colorlinks=true,citecolor=blue,pdfpagemode=UseNone}

\usepackage{amsfonts}
\usepackage{amssymb}
\usepackage{amsmath}
\usepackage{latexsym}
\usepackage{amsthm}
\usepackage{sectsty}
\usepackage[usenames]{color}
\usepackage{authblk} % nice auth & affil
\usepackage{algorithm}
\usepackage{algorithmic}

\newtheorem{theorem}{Theorem}
\newtheorem{lemma}[theorem]{Lemma}

\theoremstyle{definition}
\newtheorem{definition}[theorem]{Definition}
\newtheorem{claim}[theorem]{Claim}

\newtheorem{remark}[theorem]{Remark}

\newcommand{\tinyspace}{\mspace{1mu}}
\newcommand{\microspace}{\mspace{0.5mu}}

\newcommand{\norm}[1]{\left\lVert\tinyspace#1\tinyspace\right\rVert}

\def\({\left(}
\def\){\right)}

\def \lket {\left|}
\def \rket {\right\rangle}
\def \lbra {\left\langle}
\def \rbra {\right|}
\newcommand{\ket}[1]{\lket\microspace #1 \microspace\rket}
\newcommand{\bra}[1]{\lbra\microspace #1 \microspace\rbra}

\newcommand{\bool}{\{0,1\}}

\newcommand{\kera}[1]{\ket{#1}\bra{#1}}
\newcommand{\bret}[2]{\langle{#1}|{#2}\rangle}
\newcommand{\veps}{\varepsilon}
\newcommand{\hybs}[2]{\ket{#1^{(#2)}}} % #1 state, #2 index;
\newcommand{\td}[2]{\textsf{td}\left(#1,#2\right)}
\DeclareMathOperator{\adv}{\mathbf{Adv}}
\begin{document}

%-----------------------------------------------------------------------------%
\title{\bf Early days following Grover's quantum search algorithm}
%-----------------------------------------------------------------------------%
\author[1]{Fang Song}
%\author[2]{Author 2}
\affil[1]{Portland State University \texttt{fang.song@pdx.edu}}
%\affil[2]{University \texttt{abc@abc}}

%\date{August 3, 2009}
\date{\today}

\maketitle

\begin{abstract}   % Short summary of paper. 
  This is a note accompanying
  \href{http://fangsong.info/teaching/s17_4510_qc/}{\textsc{``CS
      410/510: Intro to quantum computing''}} I taught at Portland
  State University in Spring 2017. It is a review and summary of some
  early results related to Grover's quantum search algorithm in a
  consistent way. I had to go back and forth among several books,
  notes and original papers to sort out various details when preparing
  the lectures, which was a pain. This is the motivation behind this
  note. I would like to thank Peter H{\o}yer for valuable feedback on
  this note.
\end{abstract}

%-----------------------------------------------------------------------------%
\section{Introduction} \label{sec:introduction}
%-----------------------------------------------------------------------------%

This note has a simple purpose to serve: to sort out some of the early
literature around Grover's quantum search algorithm and to put them in
context. If you note something interesting and/or significant that is
worth adding or something incorrect here, please drop me an email. In
the following, I will describe

\begin{itemize}
\item Quantum search algorithm when there are multiple marked items
  with or without the number of marked items
  known. (Section~\ref{sec:msearch})
\item Quantum counting, a nice application of Kitaev's phase
  estimation technique. This also gives an alternative approach to
  searching with {unknown} number of
  solutions. (Section~\ref{sec:aaqc})
\item A hybrid argument showing the optimality of Grover's algorithm
  (based on~\cite{BBBV97}). This is pretty standard, but I want to
  point out a minor tweak that shows a stronger claim: unstructured
  search, in fact the \emph{decision} version, is hard even on
  \emph{average}, strengthening the worst-case hardness which is
  commonly seen in the literature. (Section~\ref{sec:lb})
\end{itemize}

\paragraph{Overview.} Here is my narrative of the time-line based on
my reading on early papers and information I collected elsewhere
(books, lecture notes, and conversations with other researchers
including some authors of these work).

I will start the story with Grover's quantum algorithm for finding a
marked item in a dataset, which achieves quadratic speedup over
optimal classical algorithms in the query model. It was first
published in STOC'96~\cite{Grover96}, a prestigious conference in
theoretical computer science, under the relatively plain title ``A
fast quantum mechanical algorithm for database search''. He then
worked out another version with the likely more popular title
``Quantum mechanics helps in searching for a needle in a haystack''
that was geared towards physicists and appeared in
Phys. Rev. Lett.~\cite{Grover97}. Note that a couple of years ago,
Shor proposed his famous quantum factorization algorithm. Kitaev on
the other continent (Russia) apparently heard about Shor's result but
was not able to get a copy of Shor's paper back then. He instead
reproduced Shor's results on his own by a different approach, based on
his elegant and powerful tool of \emph{phase estimation}~\cite{Kit95},

Several papers followed Grover's work immediately. Boyer, Brassard,
H{\o}yer, and Tapp (these names will appear many times) first gave a
finer analysis of Grover's algorithm. They also extended it to the
setting of \emph{multiple} marked items for both cases that the number
of solutions are \emph{known} and
\emph{unknown}~\cite{BBHT96}. Applications of Grover's algorithm also
came out soon, such as quantum algorithms finding the minimum in a
dataset~\cite{DH96} and for finding collisions~\cite{BHT97}. Another
important application, \emph{quantum counting}, was already outlined
in~\cite{BBHT96}, and its full description appeared a little
later~\cite{BHT98}. The algorithm combines Grover's search algorithm
and Shor's order finding algorithm in a clever way, but the analysis
was a bit complicated. They also phrased Grover's search algorithm in
the more general framework of \emph{amplitude amplification}, which
was formalized in an earlier work by Brassard and
H{\o}yer~\cite{BH97_exact}.

Meanwhile, Cleve, Ekert, Macchiavello and Mosca were unwrapping and
extending Kitaev's work, and they rephrased many early algorithms
(Deutsch, order finding, etc.) under the \emph{phase estimation}
framework~\cite{CEMM98}. It turned out that quantum counting can be
easily described and understood under the phase-estimation framework
(in fact the proceedings version of~\cite{BHT98} already indicated
this easier perspective without providing a complete analysis in their
original approach). With Mosca, Brassard, H{\o}yer and Tapp completed
a full version~\cite{BHMT02}. This basically culminated the early
stage of developments related to quantum search.

As to the \emph{optimality} of Grover's search algorithm, an early
manuscript in 1994 by Bennett, Bernstein, Brassard and
Vazirani~\cite{BBBV97} actually preceded Grover's paper. They
systematically studied the strengths and limits of quantum
algorithms. In particular, they employed a nice tool - \emph{hybrid
  argument} to prove a lower bound on the necessary number of quantum
queries to solve the search problem with bounded error. They also
proposed using majority voting to amplify the success probability of a
quantum algorithm, but this does not achieve the quadratic speedup in
quantum amplitude amplification in~\cite{BHMT02}. Apparently, Grover
was not aware of this work by Bennett et al. when he worked out his
quantum search algorithm. \cite{BBHT96} noticed that the constant in
the lower bound of~\cite{BBBV97} was lower than that in Grover's
algorithm, and they improved it to an almost tight bound. Exact
optimality was later proved by Zalka~\cite{Zalka99} (see~\cite{DH09}
for a more intuitive proof).

Now let's get to the technical meat.

%-----------------------------------------------------------------------------%
\section{Search with multiple marked items} \label{sec:msearch}
%-----------------------------------------------------------------------------%

Let's set up Grover's search problem in the standard way. The
presentation here is adapted from Watrous's lecture
note\footnote{\url{https://cs.uwaterloo.ca/~watrous/CPSC519/LectureNotes/12.pdf}}. Let
$f:\bool^n \to \bool$ be a function\footnote{We consider $\bool^n$ for
  the sake of simplicity. For a domain $[N]$, we will need quantum
  Fourier transform on $Z_N$ rather than the simple Hadamard gate.}
and let $N = 2^n$. We are given access to the oracle $O_f$ that
implements $f$ as an unitary as usual:
\begin{equation*}
  \ket{x}\ket{y}\stackrel{O_f}{\mapsto} \ket{x}\ket{y\oplus f(x)} \, .
\end{equation*}

Define the two sets representing the ``marked'' and ``unmarked''
items:

\begin{equation*}
  A:= \{x\in\bool^n: f(x) = 1\}, \quad B: = \{x \in \bool^n: f(x) =
  0\} \, .
\end{equation*}

Then let $a = |A|$ and $b = |B|$ and define two orthonormal states:
\begin{equation*}
  \ket{A}: = \frac{1}{\sqrt a} \sum_{x\in A} \ket{x}, \quad \ket{B} :
  = \frac{1}{\sqrt b} \sum_{x\in B} \ket{x} \, .
\end{equation*}

Grover's algorithm start by preparing a state in uniform superposition
\begin{equation*}
  \ket{h} = H^{\otimes n} \ket{0^n} = \sum_{x \in \bool^n}\ket{x} \, .
\end{equation*}
Two \emph{reflection} operations are at the heart of Grover's
algorithm:
\begin{equation*}
  R_h= -H^{\otimes n} Z_0 H^{\otimes n} \quad \&  \quad R_B:= Z_f 
\end{equation*}
where
\begin{equation*}
  Z_0 \ket{x} :=\left\{
    \begin{array}{l l}
      -\ket{x} & \text{ if } x = 0^n\\
      \ket{x} & \text{ if } x\neq 0^n 
    \end{array} \right.
  \quad \&  \quad Z_f\ket{x}:= (-1)^{f(x)} \ket{x} \, . 
\end{equation*}
Note that $Z_f$ is just the ``phase''-oracle of $f$ that computes
$f(x)$ in the phase. It can be implemented from $O_f$ with one
auxiliary qubit by the standard ``phase-kick-back'' trick. $Z_0$ can
be implemented in a similar fashion (i.e., phase oracle for function
$\delta_{0^n}$.)

Consider the two dimensional plane spanned by $\{\ket{A},
\ket{B}\}$. Observe that on this plane
\begin{itemize}
\item $R_h = -H^{\otimes n} Z_0 H^{\otimes n}$: reflects a state about
  $\ket{h}$. This can be seen by noting that $H$ brings
  $\{\ket{h},\ket{h^\perp}\}$ to
  $\{\ket{0}, \ket{0^\perp} : = H^{\otimes n} \ket{h^\perp}\}$, where
  $\ket{h^\perp}$ denotes the state orthogonal to $\ket{h}$ on the
  plane of $\{\ket{A},\ket{B}\}$, and $ - Z_0$ flips the sign of all
  standard basis vectors except $\ket{0^n}$. Algebraically,
  $R_h = H^{\otimes n}(- Z_0) H^{\otimes n} = 2\kera{h} - I$. Grover
  interpreted $R_h$ as ``inversion about the mean'', which gives an
  intuitive explanation how the amplitute on the marked item
  grows. $R_h$ is sometimes also called the Grover \emph{diffusion}
  operator, probably preferred by physicists.
\item $R_B = Z_f$: reflects a state about $\ket{B}$ in the plane
  defined by $\{\ket{A},\ket{B}\}$.
\end{itemize}

The composition $G: = R_fR_B = - H^{\otimes n} Z_0 H^{\otimes n} Z_f$
thus rotates a state towards $\ket{A}$ by $2\theta_a$ (Exercise:
verify it pictorially) where
$\theta_a : = \sin^{-1}(\bret{h}{A}) = \sin^{-1}(\sqrt{\frac{a}{N}})$.

{Grover's quantum search algorithm} is then easy to describe

\begin{algorithm}
  \caption{{Grover's quantum search algorithm (number of solutions
      \textbf{known})}}
\label{alg:grover}
\begin{algorithmic}[1]

  \REQUIRE $O_f$ with $f(x) = 1$ iff. $x\in A$. $\lambda = 6/5$. 
  
  \ENSURE $x\in A$, a marked item. 

  \STATE Initialization: $\ket{h}:= H^{\otimes n} \ket{0^n} = \sum_{x \in
    \bool^n} \ket{x}$.

  \STATE \textbf{Iteration}: apply $G = (-H^{\otimes n} Z_0 H^{\otimes
    n}) Z_f$ on $\ket{h}$
  $k$ times (the number of iterations $k$ is crucial and we will
  specify it later). %Call this iteration $\gsi$.

  \STATE Measure and obtain candidate solution $x$.

\end{algorithmic}
\end{algorithm}

The effect of $G$ gives the very intuitive geometric interpretation of
Grover's algorithm: starting from $\ket{h}$, each iteration (i.e.,
application of $G$) rotates the current state by $2\theta_a$ towards
$\ket{A}$. After sufficiently many iterations, we'd hope that we are
close to $\ket{A}$ enough so we are likely to measure an element in
$A$, More precisely
\begin{eqnarray*}
  \ket{h} &=& \sin\theta_a \ket{A} + \cos\theta_a \ket{B}, \\
  G^k \ket{h} &=& \sin((2k+1)\theta_a) \ket{A} + \cos((2k+1)\theta_a)
  \ket{B}  \, .
\end{eqnarray*}

How many iterations are sufficient? If we know $a$, the number of
solutions, then it is easy to decide. We have
\begin{equation*}
  \gamma_k:=\Pr[\text{finding an $x\in A$ after $k$ iterations}] =
  \left|\bra{A}G^k\ket{h}\right|^2 = \sin^2((2k+1)\theta_a) \, .
\end{equation*}

We would like to have $(2k+1)\theta_a$ as close to $\pi/2$ as
possible. Let us pick
\begin{equation*}
k^*:= \lfloor \frac{\pi/2 - \theta_a}{2\theta_a}\rfloor \, ,
\end{equation*}
and let $\eta: = \pi/2 - (2k^*+1)\theta_a$ ($|\eta| \leq
\theta_a$). Then one finds an $x\in A$ successfully with probability
at least
\begin{equation*}
  \gamma_k = \sin^2((2k+1)\theta_a) = \sin^2(\pi/2 - \eta) = \cos^2(\eta) \geq \cos^2(\theta_a) \geq
  1- \frac{a}{N} \, .
\end{equation*}

If we just repeat the entire algorithm a few times, we can amplify the
success probability close to 1. Hence the number of queries we need is
\begin{equation*}
  O(k^*) = O(\frac{\pi}{4\theta_a}) \leq O(\frac{1}{\sin\theta_a}) =
  O(\sqrt{\frac{N}{a}}) \,.
\end{equation*}

%-----------------------------------------------------------------------------%
\subsection{Number of marked items unknown}
\label{sec:unknown}
%-----------------------------------------------------------------------------%

We can pick $k$ according to $\theta_a$ when we know $a$. But what if
we do not know $a$, the number of marked items? This is answered
in~\cite{BBHT96}. We need a simple but crucial lemma.

\begin{lemma} Let $\theta_a$ be as before (i.e.,
  $\sin^2(\theta_a) = a/N$. Let $m$ be an integer and pick
  $k\gets \{0,\ldots, m-1 \}$ uniformly at random. Then after
  applying $G$ on $\ket{h}$ $k$ times, the probability of measuring an
  $x\in A$ is
  \begin{equation*}
    P_m    = \frac{1}{2} - \frac{\sin(4m\theta_a)}{4m\sin(2\theta_a)}
    \, .
  \end{equation*}
  In particular, when $m \geq 1/{\sin(2\theta_a)}$,
  $P_m\geq \frac 1 4$.
  \label{lemma:rit}
\end{lemma}

This inspires a simple trick. We start from $m=1$, and slowly but
exponentially increment $m$, so that we can reach the right region for
$m$ without spending too many unnecessary queries\footnote{This is
  reminiscent of the \emph{exponential back-off} algorithm in some
  network protocols, but the rate of increment here needs more
  vigilance.}. In the algorithm below, we assume that $a \leq
N/2$. When there are more solutions, it is easy to sample classically
to find a marked item.

\begin{algorithm}
\caption{Quantum search with number of solutions \textbf{unknown}}
\label{alg:gunknown}
\begin{algorithmic}[1]
  \REQUIRE $O_f$ with $f(x) = 1$ iff. $x\in A$. $\lambda = 6/5$. 
  
  \ENSURE $x\in A$, a marked item. 

  \STATE Initialize $m=1$. 

  \WHILE{$m \leq \sqrt N$}
  
  \STATE pick uniformly random $k \gets \{1,\dots, m\}$.

  \STATE apply $k$ times the basic Grover iteration $G$ on initial state
  $\ket{h} = \sum_{x}\frac{1}{\sqrt N}\ket{x}$.

  \STATE measure and obtain $x$. If $x\in A$, output $x$ and
  abort. Otherwise set $m \gets \lambda m$.
    
  \ENDWHILE

\end{algorithmic}
\end{algorithm}

\begin{theorem} Algorithm~\ref{alg:gunknown} finds an $x\in A$ in
  $O(\sqrt{N/a})$ expected number of iterations.
\label{thm:gunkknown}
\end{theorem}

\begin{proof}
  Let $m^* = 1/{\sin(2\theta_a)}$ denote the critical point. Then
  \begin{equation*}
    m^* = \frac{1}{2\sin\theta_a\cos\theta_a} =
    \frac{1}{2\sqrt{a/N}\sqrt{1 - a/N}} = \frac{N}{2\sqrt{(N-a)a}} < \sqrt{N/a} \, ,
  \end{equation*}
  assuming $a \leq N/2$. Let $t = \lceil \log_{\lambda}m^*\rceil$ be
  the number of main loops in the algorithm needed to reach
  $m^*$. From Lemma~\ref{lemma:rit}, once we go beyond $m^*$, every
  loop will succeed with probability at least $1/4$. Thus we just need
  to count the number of Grover iterations necessary to reach the
  critical point plus the number of iterations afterwards to find a
  solution.

  \begin{itemize}
  \item To reach the critical point, the number of iterations in the
    $j$th loop is bounded by $m_j = \lambda^{j-1}$. Hence the total
    number of Grover iterations $G$ is bound by
    \begin{equation*}
      \sum_{j = 1}^{t} \lambda^{j-1}  = \frac{\lambda^{t} - 1}{ \lambda
        - 1} \leq \frac{\lambda}{\lambda - 1} m^* = 6m^*\, .
    \end{equation*}
  \item After reaching the critical point, let $X$ be the random
    variable denoting the Grover iterations needed till finding a
    solution. Let $L$ be the random variable denoting the additional
    loops till a solution is found. By Lemma~\ref{lemma:rit},
    $\Pr[L = \ell] \leq (\frac{3}{4})^{\ell-1}\cdot \frac{1}{4} \leq
    (\frac{3}{4})^{\ell}$.
    \begin{eqnarray*}
      E[X] &=& E_LE_X[X| L = \ell] \\
           &=& \sum_{\ell} E[X|L=\ell]\cdot \Pr[L=\ell] \\
           &\leq& \sum_{\ell}  (\sum_{j=1}^{\ell}
                  \lambda^{t+j}) \cdot \Pr[L=\ell]  \quad \text{(each loop runs $\leq
                  \lambda^{t+j}$ iterations)}\\
           & = & \sum_{\ell} \lambda^t \cdot \lambda \cdot
                 \frac{\lambda^{\ell} - 1}{\lambda -1} \cdot \Pr[L=\ell] \\
           &\leq& \frac{\lambda \cdot m^*}{\lambda - 1} \cdot
                  \sum_{\ell} \Pr[L = \ell] \cdot
                  {\lambda^\ell}\\
           &\leq& \frac{\lambda m^*}{\lambda - 1} \cdot
                  \sum_{\ell = 1}^\infty (\frac{3\lambda}{4})^\ell \\
           &\leq& \frac{\lambda \cdot m^*}{\lambda - 1} \cdot 
                  \frac{1}{1 - 3\lambda/4}  \quad \text{(b.c. we picked
                  $\lambda = 6/5 < 4/3$)}\\
           &\leq & 10 m^* \, .
    \end{eqnarray*}
   \end{itemize}

   Therefore the expected total number of Grover iterations is at most
   $6m^* + 10 m^* = O(m^*) = O(\sqrt{N/a})$.
\end{proof}

\begin{remark} Note that any $1<\lambda<4/3$ would work. Another
  feature of these algorithms is that the solution from measuring the
  register after appropriate number of Grover iterations is
  distributed \emph{uniformly} in the set of solutions. This is a key
  property behind some applications such as finding the
  minimum~\cite{DH96}.
\end{remark}

%-----------------------------------------------------------------------------%
\section{Quantum amplitude amplification and quantum counting}
\label{sec:aaqc}
%-----------------------------------------------------------------------------%

An immediate generalization leads to a general technique called
\emph{amplitude amplification}, first introduced in~\cite{BH97_exact},
but for page limitation, little details were provided. It was later
fully specified in~\cite{BHT98,BHMT02}. It is a quantum analogue of
amplifying the success probability of a randomized algorithm
classically. If we repeat independently $t$ times a randomized
algorithm that succeeds with probability $p$, then the probability
that it succeeds at least once is roughly boosted to
$1-(1-p)^t\approx tp$. Therefore we need $O(1/p)$ repetitions to
succeed with probability close to 1.

Quantum amplitude amplification takes a (classical or quantum)
subroutine $U$ that succeeds with probability $p$ (or amplitude of
magnitude $\sqrt p$), and boosts the success probability close to 1
within $O(1/{\sqrt p})$ invocations of the original subroutine. Hence
this offers a generic \emph{quadratic} speedup.

The procedure is similar to Grover's algorithm by composing two
reflections that effectively moves towards the ``good'' state. For
instance, consider a unitary operation $U$, and
$\chi: \bool^n \to \bool$. Then $\chi$ induces a partition on
$\bool^n$: call $A:=\{x\in\bool^n: \chi(x) = 1\}$ the ``good''
subspace with size $a = |A|$, and $B:=\{x\in\bool^n: \chi(x) = 0\}$
the ``bad'' subspace. Let
$\ket{A} = \frac{1}{\sqrt{|A|}}\sum_{x\in A}\ket{x}$ and
$\ket{B} = \frac{1}{\sqrt{|B|}}\sum_{x\in B}\ket{x}$. Suppose
$\ket{\psi} := U\ket{0^n} = {\sqrt p} \ket{A} +
{\sqrt{1-p}}\ket{B}$. Then define
\begin{eqnarray*}
  G:= R_A R_B  \text{ with } R_A := -U Z_0 U^* \quad \& \quad R_B :=
                                                       Z_\chi \, , 
\end{eqnarray*}
where
\begin{equation*}
  Z_0 \ket{x} :=\left\{
    \begin{array}{l l}
      -\ket{x} & \text{ if } x = 0^n\\
      \ket{x} & \text{ if } x\neq 0^n 
    \end{array} \right.
  \quad \&  \quad Z_\chi\ket{x}:= (-1)^{\chi(x)} \ket{x} \, , 
\end{equation*}
are as before. We can see that $R_A$ is a reflection about
$\ket{\psi} = U\ket{0}$, and $R_B$ is reflection about $\ket{B}$ in
the plane defined by $\{\ket{A},\ket{B}\}$, and repeated application
of $G$ on $\ket{\psi}$ will amplify the amplitude on $\ket{A}$.

\begin{remark}
  If $p$ is unknown, similar idea as in Algorithm~\ref{alg:gunknown}
  can be adapted here. Details can be found in~\cite{BHMT02}.  There
  they also described, when $p$ is known, two approaches that amplify
  the success probability to exactly one, but there is technicality
  about implementing some unitary exactly and one needs to be careful
  about the quantum gate set to work with.

  Recently, an \emph{oblivious} quantum amplitude amplification
  technique has been developed~\cite{BCCKS14,Watrous09}, which is
  applicable to a unitary with an arbitrary and \emph{unknown} input
  state $\ket{\psi}$ (rather than just $\ket{0}$).
\end{remark}

%-----------------------------------------------------------------------------%
%\subsection{Quantum amplitude amplification}
%\label{sec:qaa}
%-----------------------------------------------------------------------------%

%-----------------------------------------------------------------------------%
\subsection{Quantum counting \& amplitude estimation}
\label{sec:qc}
%-----------------------------------------------------------------------------%

Consider the same setup as in Grover's search problem, can we find out
how many marked items are there? Namely given $O_f$, can we compute
$|A| = |f^{-1}(1)|$? This is the quantum counting problem, and it can
be solved by considering a slightly more general problem,
\emph{amplitude estimation}.

Again, consider a unitary operation $U$, and
$\chi: \bool^n \to \bool$. Then $\chi$ induces a partition on
$\bool^n$: call $A:=\{x\in\bool^n: \chi(x) = 1\}$ the ``good''
subspace with size $a = |A|$, and $B:=\{x\in\bool^n: \chi(x) = 0\}$
the ``bad'' subspace. Let
$\ket{A} = \frac{1}{\sqrt a} \sum_{x \in A} \ket{a}$ and
$\ket{B} = \frac{1}{\sqrt{N-a} } \sum_{x \in B} \ket{x} $.

% Let $\Pi_A: = \sum_{x\in A}\kera{x}$ and
% $\Pi_B: = \sum_{x\in B}\kera{x}$ be the projection on the good
% subspace and bad space respectively.
  
\begin{definition}[Amplitude estimation] Let
  $\ket{\psi} = U\ket{0} = \sqrt{\alpha}\ket{A} + \sqrt{1 -
    \alpha}\ket{B}$ Estimate $\alpha$, i.e., compute $\tilde \alpha$
  such that $|\alpha - \tilde \alpha| \leq \veps$.
\end{definition}

Quantum counting is then a special case of amplitude estimation.  Let
$U = H^{\otimes n}$. Then
$\ket{h} = \frac{1}{\sqrt N}\sum_{x} \ket{x} = \sqrt{\frac{a}{N}}
\ket{A} + \sqrt{\frac{N-a}{N}} \ket{B}$. Thus $\alpha = a/N$ and a
fine estimation of $\alpha$ also gives a good approximation of $a$ -
the number of marked elements. This immediately gives an alternative
approach to searching withouth knowing the number of solutions: one
just starts off approximating the number of solutions, and proceeds
using the approximate number to decide the proper number of Grover
iterations.

So how do we solve amplitude estimation? Kitaev's \emph{phase
  estimation} technique turns out to be the bomb. The key is to
observe that the operator
\begin{equation*}
  G: = - AZ_0 A^* Z_\chi  
\end{equation*}
has eigenvectors
\begin{equation*}
  \ket{\psi_{\pm}} := \frac{1}{\sqrt 2} \left(\frac{1}{\sqrt \alpha}
    \ket{\psi_A} \pm \frac{i}{\sqrt{1-\alpha}} \ket{\psi_B}\right) \, ,
\end{equation*}
with eigenvalues $\lambda_\pm = e^{\pm i 2\theta_\alpha}$, where
$\sin^2(\theta_\alpha) = \alpha$. (Exercise: verify this.) Note that
$\ket{\psi_\pm}$ form an eigen-basis.

Kitaev's phase estimation (PE) algorithm computes an approximation of
the eigenvalue. Namely
\begin{equation*}
  \text{Input: } (U,\ket{\phi} \text{ with } U\ket{\phi} = e^{i\theta} \ket{\phi})
  \rightarrow \framebox{PE} \rightarrow \tilde \theta \approx \theta \, . 
\end{equation*}
Therefore if we can prepare any one of $\ket{\psi_\pm}$ and send it in
PE algorithm with $G$, we will be able to approximate $\theta_\alpha$
and hence $\alpha$. But how to prepare $\ket{\psi_\pm}$?

Well, we don't have to.  The trick is to note that
$U\ket{0} = \ket{\psi} = \frac{-i}{\sqrt 2} (e^{i\theta_\alpha}
\ket{\psi_+} - e^{-i\theta_\alpha}\ket{\psi_-})$ can be spanned under
the eigenvectors. Therefore if we send $\ket{\psi}$ and $G$ in the PE
algorithm, the effect will be as if measuring $\ket{\psi}$ under the
eigenbasis and then estimating the eigenvalue (phase) of the
eigenvector corresponding to the measurement outcome. Be it $\psi_+$
or $\psi_-$, we get approximation of $\alpha$ either way. Read more
details in~\cite{BHMT02}.

%-----------------------------------------------------------------------------%
\section{Optimality of Grover's algorithm: an average-case lower
  bound} \label{sec:lb}
%-----------------------------------------------------------------------------%

We discuss optimality of Grover's algorithm, i.e., hardness of solving
the search problem in this section. Read~\cite{BBBV97} for more
details such as the \emph{hybrid argument}.

Let $O_r$ define an instance of Grover's search problem where
$O_r(x) = 1$ iff. $x= r$. Given quantum access to $O_r$, an arbitrary
$k$-query quantum algorithm can be described as
\begin{equation*}
  U_k(O\otimes I_A) U_{k-1}\ldots U_1(O\otimes I_A) U_0
  \ket{0^n}\ket{0^{m-n}}_A \, ,
\end{equation*}
where $A$ is a register (work space) of $m-n$ auxiliary qubits, and
$U_j$ are arbitrary unitary operations on $m$ qubits. In what follows,
we abuse notation and write $O_r$ to represent $O_r\otimes I_A$ with
some implicit auxiliary system. 

% For simplicity, we do not
% consider the auxiliary system and it is easy to verify that all
% arguments carry over to the general case.

\subsection{Standard lower bound proof}
\label{sec:stdlb}

We will compare two states
\begin{eqnarray*}
  \ket{\psi_r^{(k)}} &:=& U_kO_rU_{k-1}\ldots U_1O_rU_0 \ket{0^m}\\
  \ket{\phi^{(k)}} &:=& U_k I U_{k-1}\ldots U_1IU_0 \ket{0^m}
\end{eqnarray*}

Identity operator $I$ can be viewed as a unitary oracle implementing
the constant-0 function, i.e., $f(x) = 0$ for all $x\in \bool^n$.

\begin{lemma} Let $N = 2^n$. $\exists r \in \bool^n$ such that
  $\norm{\hybs{\psi_r}{k} - \hybs{\phi}{k}} \leq 2k/{\sqrt N}$.
\label{lemma:lb}  
\end{lemma}

% In fact, almost the same analysis reveals a stronger claim in the
% average case.

% \begin{lemma}
%   For a randomly chosen $r\gets_R\bool^n$,
%   $\| \hybs{\psi_r}{k} - \hybs{\phi}{k}\| \leq 2k/{\sqrt N}$.
% \label{lemma:lba}  
% \end{lemma}

These immediately give hardness of the unstructured search problem
(and hence optimality of Grover's algorithm).

\begin{theorem} Any algorithm needs $\Omega(\sqrt N)$ queries to $O_r$
  in order to find $r$ with constant probability.
\label{thm:stdlb}
\end{theorem}

\begin{proof}
  For an arbitrary $k$-query algorithm $A$, it needs to be able to
  distinguish $\hybs{\psi_r}{k}$ from $\hybs{\phi}{k}$ (this will made
  more precise in the next section). Therefore
  $ \norm{\hybs{\psi_r}{k} - \hybs{\phi}{k}}$ has to be bigger than
  some constant, which implies that $k$ needs to be at least
  $\Omega(\sqrt N)$.
\end{proof}

Now let's prove Lemma~\ref{lemma:lb} by the elegant \emph{hybrid
  argument}.
\begin{proof}[Proof of Lemma~\ref{lemma:lb}]
  Introduce two sequences of intermediate states:
 \begin{eqnarray*}
   \hybs{\psi_r}{0} = U_0 \ket{0^m}, &\quad&    \hybs{\phi}{0} = U_0
                                             \ket{0^m} \\
   \hybs{\psi_r}{1} = U_1O_r \hybs{\psi_r}{0}, &\quad&
                                                       \hybs{\phi}{1} = U_1I \hybs{\phi}{0} \\
                                     &\vdots&\\
   \hybs{\psi_r}{j+1} = U_{j+1}O_r \hybs{\psi_r}{j}, &\quad&    \hybs{\phi}{j+1} = U_{j+1}I \hybs{\phi}{j} \\
                                     &\vdots&\\
   \hybs{\psi_r}{k} = U_kO_r \hybs{\psi_r}{k-1}, &\quad&    \hybs{\phi}{k} = U_kI \hybs{\phi}{k-1} \\    
  \end{eqnarray*}

Define:
\begin{equation*}
  D_r^j := \norm{\hybs{\psi_r}{j} - \hybs{\phi}{j}}, \quad  E_r^j := \norm{O_r\hybs{\phi}{j} - \hybs{\phi}{j}} \, .  
\end{equation*}
Let $\Pi_r = \kera{r}$ be the projection on the basis $\ket{r}$ (more
precisely $\Pi_r = \kera{r}\otimes I_A$ is identity on the remaining
working space).
\begin{claim} For $j = 0,\ldots, k-1$,
  $D_r^{j+1} \leq D_r^{j} + E_r^{j}$, and
  $E_r^j \leq 2 \norm{\Pi_r \hybs{\phi}{j}}$.
  \label{claim:de}
\end{claim}

\begin{proof}[Proof of Claim~\ref{claim:de}]
\begin{eqnarray*}
  D_r^{j+1} &=& \norm{ \hybs{\psi_r}{j+1} - \hybs{\phi}{j+1}}\\
            &=& \norm{U_{j+1}O_r \hybs{\psi_r}{j} - U_{j+1}I\hybs{\phi}{j}}\\
            &=& \norm{U_{j+1}O_r \hybs{\psi_r}{j} - U_{j+1}O_r \hybs{\phi}{j} 
                + U_{j+1}O_r \hybs{\phi}{j} - U_{j+1} \hybs{\phi}{j}}\\
            &\leq& \norm{U_{j+1}O_r(\hybs{\psi_r}{j} - \hybs{\phi}{j})} \\ 
            & + & \norm{U_{j+1}(O_r\hybs{\phi}{j} - \hybs{\phi}{j})} \quad \text{triangle-inequality}\\
            &=&  \norm{\hybs{\psi_r}{j} - \hybs{\phi}{j}} +
                \norm{O_r\hybs{\phi}{j} - \hybs{\phi}{j}} \quad
                \text{unitary preserves norm} \\
            & = & D_r^{j} + E_r^{j} \, .
\end{eqnarray*}

For the second part, note that
\begin{eqnarray*}
  O_r\hybs{\phi}{j}  &=& O_r(\Pi_r+(I - \Pi_r)) \hybs{\phi}{j-1}\\
                     &=& O_r\Pi_r \hybs{\phi}{j-1} + (I-\Pi_r) \hybs{\phi}{j-1}
                         \quad \text{(Note: $O_r(I-\Pi_r) = I - \Pi_r$)} \, ; \\
  \hybs{\phi}{j-1} &=& \Pi_r\hybs{\phi}{j-1} + (I -
                       \Pi_r)\hybs{\phi}{j-1} \, .
\end{eqnarray*}

Therefore
\begin{equation*}
  E_r^j  = \norm{ O_r\Pi_r \hybs{\phi}{j-1} - \Pi_r\hybs{\phi}{j-1}}
  \leq \norm{O_r\Pi_r \hybs{\phi}{j-1}}+ \norm{\Pi_r\hybs{\phi}{j-1}} =
  2 \norm{\Pi_r\hybs{\phi}{j-1}} \, .
\end{equation*}

\end{proof}

Back to proving the Lemma, we have 
%\begin{equation*}
$ \norm{ \hybs{\psi_r}{k} - \hybs{\phi}{k}} = D_r^k \leq E_r^{k-1} +
\ldots + E_r^{0} = \sum_{j=0}^{k-1} E_r^{j} \, .$
%\end{equation*}

\begin{claim} 
  $\sum_{r\in \bool^n}D_r^k \leq 2k\sqrt N$.
  \label{claim:sumd}
\end{claim}

\begin{proof}[Proof of Claim~\ref{claim:sumd}]

\begin{eqnarray*}
  \sum_{r\in\bool^n} D_r^k=\sum_r \sum_j E_r^j &=& \sum_j \sum_r E_r^j \\
                                               &\leq& \sum_j \sqrt{N
                                                      \sum_{r}
                                                      (E_r^j)^2} \quad \text{(Cauchy-Schwarz)} \\
                                               &\leq& 2 \sqrt N \sum_j
                                                      \sum_{r} \norm{\Pi_r
                                                      \hybs{\phi}{j-1}}^2
                                                      = 2 \sqrt{N}
                                                      \sum_{j=0}^{k-1}
                                                      \norm{\hybs{\phi}{j-1}}^2 \\
                                               &=& 2 \sqrt{N}
                                                   \sum_{j=0}^{k-1} 
                                                   1 = 2k \sqrt{N} \, .
\end{eqnarray*}

\end{proof}
Therefore there must be at least one $r^*$ such that
$D_{r^*}^k \leq 2k/{\sqrt N}$, because otherwise the
$\sum_r D_r^k > N \cdot 2k/{\sqrt N} = 2k \sqrt N$.
\end{proof}

\subsection{Stronger lower bound: average-case hardness}
\label{sec:avglb}
Note that the above proof actually holds for the \emph{decision}
version: for some $r$, distinguishing $O_r$ from constant-0 function
is hard. Namely, deciding if an oracle contains a marked item is
already hard, which of course implies that finding a solution is at
least as hard. In fact, we can further observe something stronger. Let
us introduce some more basic notions to make a formal statement.

Recall the trace distance
$\td{\rho}{\sigma}: = \frac{1}{2}\norm{\rho - \sigma }_1 =
\frac{1}{2}Tr(\sqrt{(\rho-\sigma)^*(\rho - \sigma)})$, where
$Tr(\cdot)$ computes the trace of a matrix. For two pure states
$\ket{\psi}$ and $\ket{\phi}$, it is easy to verify that
\begin{equation}
  \td{\ket{\psi}}{\ket{\phi}} = \sqrt{1 - |\bret{\psi}{\phi}|^2} \leq
  \norm{\ket{\psi} - \ket{\phi}}\, .
  \label{eqn:tdnorm}
\end{equation}
  %  Lemme~\ref{lemma:lba} tells us that
  % $\| \hybs{\psi_r}{k} - \hybs{\phi}{k}\| \leq 2k/{\sqrt N}$. This
  % implies that
  % \begin{equation*}
  %   \td{\hybs{\psi_r}{k}}{\hybs{\phi}{k}} \leq 2k/{\sqrt N}   
  % \end{equation*}

Consider two oracles $O$ and $O'$. Let $A^O(\cdot)$ denote an
algorithm $A$ that makes queries to $O$ and finally output one bit. We
define the distinguishing \emph{advantage} of an algorithm $A$ trying
to tell apart $O$ and $O'$:
\begin{equation*}
  \adv_A^{O,O'} := \lket \Pr[A^{O}(\cdot)=1] - \Pr[A^{O'}(\cdot) =1] \rbra \, .
\end{equation*}
We will be interested in distinguishing the oracle $O_r$ and identity
$I$ (i.e., constant-0 function).

Consider the process that a random $r$ is chosen, and then $O_r$ is
given to an algorithm $A$. The goal is to find $r$, the marked
element. We show that this is hard.
\begin{theorem} For any $k$-query algorithm $A$,
\begin{equation*}
  \adv_A^{O,I} : = \lket\Pr_{r\gets\bool^n}[A^{O_r}(\cdot) =1] -
  \Pr[A^{I}(\cdot) =1]\rbra  \leq 2k/{\sqrt N} \, .
\end{equation*}
Therefore one needs at least $\Omega(\sqrt N)$ queries to find the
marked element.
\label{cor:groverlba}
\end{theorem}

This theorem is strong in a couple of aspects.

\begin{itemize}
\item It holds for a random Grover oracle, not just in the
  worst-case. This theorem also explicitly refers to the
  \emph{decision} version of unstructured search problem. It follows
  that, for a \emph{uniformly} random chosen $r$, finding $r$ is as
  hard as the worst case.
\item We bound the success probability of any algorithm with certain
  number of queries. This is most relevant in the cryptographic
  setting, since even a small (e.g. inverse polynomial) winning
  probability matters. Usually in the literature, one only cares about
  the hardness for solving the problem with constant probability.
\end{itemize}

Exercise: prove the classical lower bound for solving this
average-case (under uniform distribution) search problem.

\begin{proof}
  
  Let $p_r: = \Pr[A^{O_r}(\cdot) =1]$ and $q:=\Pr[A^{I}(\cdot) =1]|$.
  By Holevo-Helstrom theorem~\cite{Hel67,Hol72}, we know
  that~\footnote{C.f. Theorem 3.4. of
    \url{https://cs.uwaterloo.ca/~watrous/TQI/TQI.3.pdf}, and set
    $\lambda = 1/2$.}
  \begin{equation}
    \adv_A^{O_r,I} = |\Pr[A^{O_r}(\cdot)=1] - \Pr[A^{I}(\cdot) =1]| \leq
    \td{\hybs{\psi_r}{k}}{\hybs{\phi}{k}}\, .
    \label{eqn:adv}
  \end{equation}
  Therefore we have that 
\begin{eqnarray*}
  \adv^{O,I} &=& \lket\Pr_{r\gets\bool^n}[A^{O_r}(\cdot) =1] -
                 \Pr[A^{I}(\cdot) =1]\rbra \\
             &=& \lket \sum_r(\Pr[A^{O_r}(\cdot) =1|r] \cdot\Pr_{r\gets
                 \bool^n}[r]) -  \Pr[A^{I}(\cdot) =1]\rbra \\
             & = & \lket \sum_r p_r\cdot \frac 1 N - q \rbra =
                   \frac{1}{N} \lket \sum_r(p_r - q)\rbra \\
             &\leq& \frac{1}{N}\sum_r |p_r - q|\\
             &\leq & \frac{1}{N}\sum_r
                     \td{\hybs{\psi_r}{k}}{\hybs{\phi}{k}} \quad
                     \text{(Eqn.~\ref{eqn:adv})} \\
             &\leq& \frac{1}{N}\sum_r
                    \norm{{\hybs{\psi_r}{k}} - {\hybs{\phi}{k}}} \quad
                    \text{(Eqn.~\ref{eqn:tdnorm})}\\
             &=& \frac{1}{N}\sum_r D_r^k \leq \frac{2k}{\sqrt N} \quad
                 \text(Claim~\ref{claim:sumd})\, .
\end{eqnarray*}
\end{proof}

Exercise: can you extend everything here to the case of multiple
solutions?

\newpage
\small{
\bibliographystyle{alpha}
\bibliography{grover}
}
% \begin{thebibliography}

% \bibitem[AB09]{AroraB09}
% S.~Arora and B.~Barak.
% \newblock {\em Computational Complexity: A Modern Approach}.
% \newblock Cambridge University Press, 2009.

% \end{thebibliography}

\end{document}